\documentclass[letterpaper, 10 pt, conference]{ieeeconf}  
\IEEEoverridecommandlockouts                              
\usepackage[utf8]{inputenc}
\usepackage{amsmath}
\usepackage{amsthm}
\usepackage{mathrsfs}
\usepackage{amssymb}
\usepackage{color}
\usepackage{cite} 
\usepackage[font=footnotesize]{caption}
\usepackage{comment}
\usepackage{float}
\usepackage{graphicx}
\usepackage{float}
\usepackage{lscape}
\usepackage{bm}
\usepackage{booktabs}
\usepackage{longtable,tabularx,ragged2e}
\usepackage[super]{nth}
\usepackage{multicol}
\usepackage{booktabs}
\usepackage{lscape}
\usepackage{nicefrac}
\usepackage{siunitx}

\newtheorem{theo}{Theorem}
\newtheorem{prop}{Proposition}
\newtheorem{cor}{Corollary}

\newtheorem{lem}{Lemma}

\newcommand{\rmT}{{\rm T}}

\newcommand{\BBR}{{\mathbb R}}

\DeclareMathOperator*{\argmin}{arg\,min}

\title{Adaptive Kalman Filtering Developed from Recursive Least Squares Forgetting Algorithms}

\author{\large Brian Lai and Dennis S. Bernstein 
\thanks{Brian Lai and Dennis S. Bernstein are with the Department of Aerospace Engineering, University of Michigan, Ann Arbor, MI, USA. {\tt \{brianlai,  dsbaero\}@umich.edu}. This work was supported by the NSF Graduate Research Fellowship under Grant No. DGE 1841052.}
}

\begin{document}
\maketitle

\begin{abstract}
Recursive least squares (RLS) is derived as the recursive minimizer of the least-squares cost function.
Moreover, it is well known that RLS is a special case of the Kalman filter.
This work presents the Kalman filter least squares (KFLS) cost function, whose recursive minimizer gives the Kalman filter. 
KFLS is an extension of generalized forgetting recursive least squares (GF-RLS), a general framework which contains various extensions of RLS from the literature as special cases.
This then implies that extensions of RLS are also special cases of the Kalman filter.
Motivated by this connection, we propose an algorithm that combines extensions of RLS with the Kalman filter, resulting in a new class of adaptive Kalman filters.
A numerical example shows that one such adaptive Kalman filter provides improved state estimation for a mass-spring-damper with intermittent, unmodeled collisions.
This example suggests that such adaptive Kalman filtering may provide potential benefits for systems with non-classical disturbances.
\end{abstract}

\section{Introduction}
Despite their respective deterministic and stochastic foundations, least-squares and the Kalman filter share an interconnected history \cite{Sorenson1970Least}. 
It is well known that the update equations for recursive least squares (RLS) (e.g. \cite{islam2019recursive}) are the same as those of the Kalman filter with, for all $k \ge 0$, identity state matrix $A_k = I$, zero input matrix $B_k = 0$, process noise covariance $\Sigma_k = 0$, and measurement noise covariance $\Gamma_k = I$ (see p.51 of \cite{aastrom2013adaptive}, section 3.3.5 of \cite{crassidis2004optimal}, or p.129 of \cite{simon2006optimal}). 
As RLS became a foundational algorithm of systems and control theory for online identification of fixed parameters \cite{ljung1983theory,aastrom2013adaptive}, numerous extensions of RLS were developed (e.g. \cite{bruce2020convergence,cao2000directional,kulhavy1984tracking,bittanti1990convergence,
Salgado1988modified,Lai2023Exponential,vahidi2005recursive,
paleologu2008robust,johnstone1982exponential,Dasgupta1987Asymptotically}) to improve identification of time-varying parameters.
However, little work has been done to connect these extensions to the Kalman filter.

The RLS update equations are often derived as the recursive minimizer of a least squares cost function (e.g. \cite{islam2019recursive}).
A natural question is whether the RLS cost function can be generalized to derive the Kalman filter.
While other deterministic derivations of the Kalman filter have been presented (e.g. \cite{Humpherys2010Kalman,humpherys2012fresh}), these do not follow as an extension of the RLS cost.

This work presents the Kalman filter least squares (KFLS) cost function whose recursive minimizer gives the Kalman filter update equations.
KFLS is an extension of generalized forgetting recursive least squares (GF-RLS) \cite{lai2023generalized}, which contains various extensions of RLS from the literature as special cases. 
As a result, these extensions of RLS are also special cases of the Kalman filter with particular choices of the process noise covariance matrix.

This result motivates a new class of adaptive Kalman filtering, with modified prior covariance update equations to incorporate forgetting from extensions of RLS. 
A brief survey is given to show how several forgetting methods from the RLS literature can be applied to adaptive Kalman filtering.
A numerical example shows how adaptive Kalman filtering with robust variable forgetting factor \cite{paleologu2008robust} improves state estimation of a mass-spring-damper system with intermittent, unmodeled collisions.
This example suggests that such an adaptive Kalman filtering may be beneficial when disturbances are non-classical.

\subsubsection{Notation and Terminology}
For symmetric $P,Q \in \BBR^{n \times n}$, $P \prec Q$ (respectively, $P \preceq Q$) denotes that $Q-P$ is positive definite (respectively, positive semidefinite).
For all $x \in \BBR^n$, let $\Vert x \Vert \triangleq \sqrt{x^\rmT x}$. 
For $x \in \BBR^n$ and positive-semidefinite $R \in \BBR^{n \times n}$, $\Vert x \Vert_R \triangleq \sqrt{x^\rmT R x}$. 

\section{Background Material}

\subsection{The one-step Kalman Filter}

Consider the discrete-time, linear, time-varying system
\begin{align}
    x_{k+1} &= A_k x_k + B_k u_k + w_k, \label{eqn: KF state} \\
    y_k &= C_k x_k + v_k, \label{eqn: KF measurement}
\end{align}
where, for all $k \ge 0$, $x_k \in \BBR^n$ is the state, $y_k \in \BBR^{p}$ is the measurement, $u_k \in \BBR^m$ is the input, $w_k \sim \mathcal{N}(0,\Sigma_k)$ is the process noise, and $v_k \sim \mathcal{N}(0,\Gamma_k)$ is the measurement noise, for positive-semidefinite $\Sigma_k \in \BBR^{n \times n}$ and positive-semidefinite $\Gamma_k \in \BBR^{p \times p}$. 
Moreover, for all $k \ge 0$, $A_k \in \BBR^{n \times n}$, $B_k \in \BBR^{n \times m}$, and $C_k \in \BBR^{p \times n}$. 
The \textit{two-step Kalman filter} \cite{Kalman1960Filter} for the system given by \eqref{eqn: KF state} and \eqref{eqn: KF measurement} is can be expressed as
\begin{align}
    \hat{x}_{k+1 \vert k} &= A_k \hat{x}_k + B_k u_k, \label{eqn: KF x predict} \\
    P_{k+1 \vert k} &= A_k P_k A_k^\rmT + \Sigma_k, \label{eqn: KF P predict} \\
     K_k &= P_{k+1 \vert k} C_k^\rmT (C_k P_{k+1 \vert k} C_k^\rmT + \Gamma_k)^{-1} \\
    \hat{x}_{k+1} &= \hat{x}_{k+1 \vert k} + K_k (y_k - C_k \hat{x}_{k+1 \vert k}), \label{eqn: KF x update} \\
    P_{k+1} &= P_{k+1 \vert k} - K_k C_k P_{k+1 \vert k}, \label{eqn: KF P update}
\end{align}
where, for all $k \ge 0$, positive-definite  $P_{k+1 \vert k} \in \BBR^{n \times n}$ and $P_k \in \BBR^{n \times n}$ are, respectively, the prior and posterior covariances, $\hat{x}_{k + 1 \vert k} \in \BBR^n$ and $\hat{x}_k \in \BBR^n$ are, respectively, the prior and posterior state estimates, and $K_k \in \BBR^{n \times p}$ is the Kalman gain.
%

Next, if, for all $k \ge 0$, $\Sigma_k + A_k P_k A_k^\rmT$ and $\Gamma_k$ are nonsingular, then the matrix inversion lemma (Lemma \ref{lem: matrix inversion lemma}) can be used to rewrite \eqref{eqn: KF x predict} through \eqref{eqn: KF P update} as the \textit{one-step Kalman filter} \cite{humpherys2012fresh}, where, for all $k \ge 0$, 
\begin{align}
    P_{k+1}^{-1} =&  (\Sigma_k + A_k P_k A_k^\rmT)^{-1}  + C_k^\rmT \Gamma_k^{-1} C_k, \label{eqn: one step KF 1} \\
    \hat{x}_{k+1} =& A_k \hat{x}_k + B_k u_k  \nonumber
    \\
    & + P_{k+1} C_k^\rmT \Gamma_k^{-1} \left(y_k - C_k (A_k \hat{x}_k + B_k u_k) \right). \label{eqn: one step KF 2}
\end{align}

\subsection{Discrete-Time LTV State Transition Function}
To facilitate expressing a least squares cost function for the Kalman filter, we first introduce the \textit{state transition function} for discrete-time LTV system, a concise notation to transition between states at different time steps. 
For all $k \ge 0$, let $A_k \in \BBR^{n \times n}$, $B_k \in \BBR^{n \times m}$, $x_k \in \BBR^{n}$, and $u_k \in \BBR^{m}$. Consider the discrete-time, linear time-varying state update equation
\begin{align}
    x_{k+1} = A_k x_k + B_k u_k. \label{eqn: LTV state update}
\end{align}
Assume, for all $k \ge 0$, $A_k$ is nonsingular. For all $i,k \ge 0$, define the \textit{state transition matrix} (from step $i$ to step $k$), denoted $\Phi_{k,i} \in \BBR^{n \times n}$, by
\begin{align}
    \Phi_{k,i} \triangleq \begin{cases}
    A_{k-1} A_{k-2} \cdots A_{i+1} A_i & i < k, \\
    I & i = k, \\
    A_{k}^{-1} A_{k+1}^{-1} \cdots A_{i-2}^{-1} A_{i-1}^{-1} & k < i.
    \end{cases}
\end{align}
It follows that, for all $i,k \ge 0$, $\Phi_{i,k}^{-1} = \Phi_{k,i}$. 
Then, for all $0 \le i < k$, $x_k$ can be expressed as
\begin{align}
    x_k &= A_{k-1}x_{k-1} + B_{k-1} u_{k-1}, \nonumber \\
    &= A_{k-1} A_{k-2} x_{k-2} + A_{k-1}B_{k-2} u_{k-2} + B_{k-1} u_{k-1}, \nonumber \\
    &= \ \cdots \ = \Phi_{k,i} x_i + \sum_{j=i}^{k-1} \Phi_{k,j+1} B_j u_j. \label{eqn: x_k forward transition}
\end{align}
For all $0 \le i < k$, we further define the matrices
\begin{align}
    \mathcal{B}_{k,i} &\triangleq \begin{bmatrix} 
        \Phi_{k,{i+1}} B_i & 
        \cdots & \Phi_{k,{k-1}} B_{k-2} & \Phi_{k,{k}} B_{k-1}
    \end{bmatrix}, 
    \\
    \mathcal{U}_{k,i} &\triangleq \begin{bmatrix}
        u_i^\rmT & 
        \cdots & u_{k-2}^\rmT & u_{k-1}^\rmT
    \end{bmatrix}^\rmT.
\end{align}
It follows that $\mathcal{B}_{k,i} \mathcal{U}_{k,i} = \sum_{j=i}^{k-1} \Phi_{k,j+1} B_j u_j$.
Hence, for all $0 \le i < k$, $x_{k} = \Phi_{k,i} x_i + \mathcal{B}_{k,i} \mathcal{U}_{k,i}$.
On the other hand, for all $0 \le k < i$, $x_k = \Phi_{k,i} (x_{i} - \mathcal{B}_{i,k} \mathcal{U}_{i,k} )$.
Therefore, for all $i \ge 0$ and $k \ge 0$, we define the \textit{state transition function} (from step $i$ to step $k$), denoted $\mathcal{T}_{k,i}: \BBR^n \rightarrow \BBR^n$, by
\begin{align}
    \mathcal{T}_{k,i}(x) \triangleq \begin{cases}
        \Phi_{k,i} x + \mathcal{B}_{k,i} \mathcal{U}_{k,i} 
        & i < k, \\
        x & i = k, \\
        \Phi_{k,i} (x  - \mathcal{B}_{i,k} \mathcal{U}_{i,k} ) 
        & k < i.
    \end{cases}
    \label{eqn: state transition function}
\end{align}

\section{A Least Squares Cost Function which Derives the Kalman Filter}
\label{sec: KFRLS}

This section develops the Kalman filter least squares (KFLS) cost function whose recursive minimizer gives the one-step Kalman filter update equations.
To begin, for all $k \ge 0$, let $F_k \in \BBR^{n \times n}$ be the \textit{forgetting matrix}.
Theorem \ref{theo: KFRLS} develops the KFLS cost function \eqref{eqn: KFRLS cost} in terms of $F_k$ and shows how the update equations \eqref{eqn: KFRLS Pinv Update} and \eqref{eqn: KFRLS x Update} minimize that cost.
Corollary \ref{cor: KFRLS and Kalman Filter} will later show how, for a particular choice of $F_k$, the update equations of Theorem \ref{theo: KFRLS} are equivalent to the one-step Kalman filter.

\begin{theo}
\label{theo: KFRLS}
For all $k \ge 0$, let $A_k \in \BBR^{n \times n}$ be nonsingular, $B_k \in \BBR^{n \times m}$, $C_k \in \BBR^{p \times n}$, $\Gamma_k \in \BBR^{p \times p}$ be positive definite, $u_k \in \BBR^{m}$, and $y_k \in \BBR^{p}$. Furthermore, let $P_0 \in \BBR^{n \times n}$ be positive definite and $\hat{x}_0 \in \BBR^{n}$. For all $k \ge 0$, Let $F_k \in \BBR^{n \times n}$ be positive semidefinite and such that
\begin{align}
    F_k & \prec \Phi_{0,k}^\rmT P_0^{-1} \Phi_{0,k} \nonumber 
    \\ 
    & + \sum_{i=0}^{k-1} \left( \Phi_{i+1,k}^\rmT C_i^\rmT \Gamma_i^{-1} C_i \Phi_{i+1,k} - \Phi_{i,k}^\rmT F_i \Phi_{i,k} \right).
    \label{eqn: Pinv - F is pos def}
\end{align}
For all $k \ge 0$, define $J_k \colon \BBR^n \rightarrow \BBR$ as
\begin{equation}
    \label{eqn: KFRLS cost}
    \begin{aligned}
        J_{k}(\hat{x}) \triangleq J_{k,{\rm loss}}(\hat{x}) - J_{k,{\rm forget}}(\hat{x}) + J_{k,{\rm reg}}(\hat{x}),
    \end{aligned}
\end{equation}
where
\begin{align}
    J_{k,{\rm loss}}(\hat{x}) &\triangleq \sum_{i=0}^{k} \Vert y_i - C_i \mathcal{T}_{i+1,k+1}(\hat{x}) \Vert _{\Gamma_i^{-1}}^2 \\ 
    J_{k,{\rm forget}}(\hat{x}) &\triangleq \sum_{i=0}^{k} \Vert \mathcal{T}_{i,k+1}(\hat{x}) - \hat{x}_i \Vert_{ F_i}^2 , \\
    J_{k,{\rm reg}}(\hat{x}) &\triangleq  \Vert \mathcal{T}_{0,k+1}(\hat{x}) - \hat{x}_0 \Vert_{P_0^{-1}}^2.
\end{align}
Then, there exists a unique global minimizer of $J_k(\hat{x})$, denoted $\hat{x}_{k+1} \triangleq \argmin_{\hat{x} \in \BBR^n} J_k(\hat{x})$,
which, for all $k \ge 0$, is given recursively by
\begin{align}
    P_{k+1}^{-1} =& A_k^{-\rmT} ( P_k^{-1} - F_k) A_k^{-1}  + C_k^\rmT \Gamma_k^{-1} C_k, \label{eqn: KFRLS Pinv Update}
    \\
    \hat{x}_{k+1} =& A_k \hat{x}_k + B_k u_k \nonumber
    \\
    &+ P_{k+1} C_k^\rmT \Gamma_k^{-1} \left(y_k - C_k (A_k \hat{x}_k + B_k u_k) \right), \label{eqn: KFRLS x Update}
\end{align}
where, for all $k \ge 0$, $P_k \in \BBR^{n \times n}$ is positive definite.
\end{theo}

\begin{proof}
    See the Appendix.
\end{proof}

\begin{cor}
\label{cor: Pkinv - Fk is pos def KFRLS}
    Consider the notation and assumptions of Theorem \ref{theo: KFRLS}. 
    For all $k \ge 0$, 
    \begin{align}
        P_k^{-1} &= \Phi_{0,k}^\rmT P_0^{-1} \Phi_{0,k} \nonumber 
    \\ 
    & + \sum_{i=0}^{k-1} \left( \Phi_{i+1,k}^\rmT C_i^\rmT \Gamma_i^{-1} C_i \Phi_{i+1,k} - \Phi_{i,k}^\rmT F_i \Phi_{i,k} \right), \label{eqn: Pkinv expanded}
    \end{align}
    and hence \eqref{eqn: Pinv - F is pos def} hold if and only if $P_k^{-1} - F_k \succ 0$.
\end{cor}
\begin{proof}
    Let $k \ge 0$. Note that \eqref{eqn: Pkinv expanded} follows from repeated substitution of \eqref{eqn: KFRLS Pinv Update}. 
    Next, note that the right-hand side of \eqref{eqn: Pkinv expanded} is equivalent to the right-hand side of \eqref{eqn: Pinv - F is pos def}. 
    Hence, \eqref{eqn: Pinv - F is pos def} is equivalent to $P_k^{-1} - F_k \succ 0$
\end{proof}

\begin{cor}
\label{cor: KFRLS and Kalman Filter}
    Consider the notation and assumptions of Theorem \ref{theo: KFRLS}. 
    If, for all $k \ge 0$, there exists positive-semidefinite $\Sigma_k \in \BBR^{n \times n}$ such that
    \begin{align}
    \label{eqn: F_k in terms of Sigma_k}
    F_k = P_k^{-1} - (A_k^{-1} \Sigma_k A_k^{-\rmT} + P_k)^{-1},
    \end{align}
    then, for all $k \ge 0$, \eqref{eqn: Pinv - F is pos def} is satisfied. Moreover, \eqref{eqn: KFRLS Pinv Update} and \eqref{eqn: KFRLS x Update} are equivalent to the one-step Kalman filter update equations \eqref{eqn: one step KF 1} and \eqref{eqn: one step KF 2}.
\end{cor}

\begin{proof}
    Let $k \ge 0$. Note that \eqref{eqn: F_k in terms of Sigma_k} can be rewritten as
    \begin{align}
    \label{eqn: KFRLS Sigma_k F_k relation}
    (\Sigma_k + A_k P_k A_k^\rmT)^{-1} = A_k^{-\rmT} (P_k^{-1} - F_k) A_k^{-1}.
    \end{align}
    Substituting \eqref{eqn: KFRLS Sigma_k F_k relation} into \eqref{eqn: KFRLS Pinv Update}, it follows that \eqref{eqn: KFRLS Pinv Update} and \eqref{eqn: KFRLS x Update} are equivalent to \eqref{eqn: one step KF 1} and \eqref{eqn: one step KF 2}.
    Moreover, \eqref{eqn: F_k in terms of Sigma_k} can also be expressed as $P_k^{-1} - F_k = (A_k^{-1} \Sigma_k A_k^{-\rmT} + P_k)^{-1}$. 
    Since $\Sigma_k$ and $P_k$ are both positive definite, it follows that $A_k^{-1} \Sigma_k A_k^{-\rmT} + P_k$ is positive definite, and hence $P_k^{-1} - F_k$ is positive definite.
    Therefore, by Corollary \ref{cor: Pkinv - Fk is pos def KFRLS}, condition \eqref{eqn: Pinv - F is pos def} is satisfied.
\end{proof}

Next, Proposition \ref{prop: Sigma_k PSD (PD)} shows that, for all $k \ge 0$, $\Sigma_k$ has the same definiteness as $F_k$.
\begin{prop}
\label{prop: Sigma_k PSD (PD)}
For all $k \ge 0$, $\Sigma_k$ is positive semidefinite (resp. positive definite) if and only on $F_k$ is positive semidefinite (resp. positive definite).
\end{prop}
\begin{proof}
    Let $k \ge 0$. 
    If $F_k$ is positive semidefinite (respectively, positive definite), then $P_k^{-1} -F_k \preceq P_k^{-1}$ (respectively, $P_k^{-1} -F_k \prec P_k^{-1}$). 
    Since $P_k^{-1} - F_k$ is nonsingular by Corollary \ref{cor: Pkinv - Fk is pos def KFRLS}, it follows that $(P_k^{-1} -F_k)^{-1} \succeq P_k$ and hence $(P_k^{-1} -F_k)^{-1} - P_k \succeq 0$ (respectively, $(P_k^{-1} -F_k)^{-1} -  P_k \succ 0$). 
    Finally, since $A_k$ is nonsingular by assumption, it follows from \eqref{eqn: Sigma_k definition} that $\Sigma_k$ is positive semidefinite (respectively, positive definite).
    Next, if $\Sigma_k$ is positive semidefinite (respectively, positive definite), then $\Sigma_k + P_k \succeq P_k$ (respectively, $\Sigma_k + P_k \succ P_k$) and hence $(\Sigma_k + P_k)^{-1} \preceq P_k^{-1}$ (respectively, $(\Sigma_k + P_k)^{-1} \prec P_k^{-1}$). 
    Therefore, $F_k = P_k^{-1} - (\Sigma_k + P_k)^{-1} \succeq 0$ (respectively, $F_k = P_k^{-1} - (\Sigma_k + P_k)^{-1} \succ 0$).
\end{proof}

%

\section{RLS Extensions as Special Cases of the Kalman Filter}
Revisiting the state update equation \eqref{eqn: LTV state update}, note that if, for all $k \ge 0$, $A_k = I_n$ and $B_k = 0_{n \times m}$, then, for all $k \ge 0$, $x_{k+1} = x_k$.
This also implies that the state transition function, given by \eqref{eqn: state transition function}, is identity.
In particular, for all $i \ge 0$, $k \ge 0$, and $x \in \BBR^n$, $\mathcal{T}_{k,i}(x) = x$.
Substituting the identity state transition function into the cost function $J_k$ given by \eqref{eqn: KFRLS cost}, it follows that, for all $k \ge 0$, 
\begin{align}
    J_{k,{\rm loss}}(\hat{x}) &= \sum_{i=0}^{k} \Vert y_i - C_i \hat{x} \Vert _{\Gamma_i^{-1}}^2 \\ 
    J_{k,{\rm forget}}(\hat{x}) &= \sum_{i=0}^{k} \Vert \hat{x} - \hat{x}_i \Vert_{ F_i}^2 , \\
    J_{k,{\rm reg}}(\hat{x}) &=  \Vert \hat{x} - \hat{x}_0 \Vert_{P_0^{-1}}^2.
\end{align}

Note that, in this special case, the cost function $J_k$ is equivalent to the generalized forgetting recursive least squares (GF-RLS) cost developed in \cite{lai2023generalized}, where GF-RLS uses the notation $\theta \in \BBR^n$ and $\phi_k \in \BBR^{p \times n}$ instead of $\hat{x} \in \BBR^n$ and $C_k \in \BBR^{p \times n}$, respectively.
In \cite{lai2023generalized}, it was shown that various extensions of RLS from the literature are special cases of GF-RLS if, for all $k \ge 0$, a particular forgetting matrix $F_k \in \BBR^{n \times n}$ is chosen.

Therefore, we conclude that an extension of RLS, which is a special case of GF-RLS with forgetting matrix $F_k$, $k \ge 0$, is also a special case of the Kalman filter if, for all $k \ge 0$, $A_k = I_n$, $B_k = 0_{n \times m}$, and there exists positive-semidefinite $\Sigma_k \in \BBR^{n \times n}$ such that \eqref{eqn: F_k in terms of Sigma_k} holds.
Explicitly solving for $\Sigma_k$, it follows that
\begin{align}
    \label{eqn: Sigma_k definition}
    \Sigma_k = A_k \left[ (P_k^{-1} - F_k)^{-1} - P_k \right] A_k^\rmT,
\end{align}
where $P_k^{-1} - F_k$ is nonsingular by \eqref{eqn: Pinv - F is pos def} and Corollary \ref{cor: Pkinv - Fk is pos def KFRLS}. 
Moreover, by Proposition \ref{prop: Sigma_k PSD (PD)}, $\Sigma_k$ is positive semidefinite if and only if $F_k$ is positive semidefinite.

While \cite{lai2023generalized} gives a thorough literature review on extensions of RLS as special cases of GF-RLS, for brevitiy, we summarize eight extensions in Table \ref{table:summary}.
Given are the algorithm name and reference, assumptions of the algorithm, tuning parameters, and $\Sigma_k \in \BBR^{n \times n}$ derived from \eqref{eqn: Sigma_k definition}.

\begin{table*}[ht]\centering
\caption{RLS Extensions as Special Cases of the Kalman Filter}
\setlength\tabcolsep{6pt} 
\renewcommand{\arraystretch}{1.1} 
\begin{tabular}{@{}lll@{}}
\toprule[1pt] 
Algorithm & Tuning Parameters & Process Noise Covariance $\Sigma_k$ 
\\
\midrule
1. Recursive Least Squares \cite{islam2019recursive} & --- & $\Sigma_k = 0_{n \times n}$
\\ 
2. Exponential Forgetting \cite{islam2019recursive,goel2020recursive} & $\lambda \in (0,1]$ & $\Sigma_k = (\frac{1}{\lambda} - 1)P_k.$ 
\\
3. Variable-Rate Forgetting \cite{bruce2020convergence} & $\lambda_k \in (0,1]$, $k \ge 0$ & $\Sigma_k = (\frac{1}{\lambda_k} - 1)P_k.$ 
\\
4. Data-Dependent Forgetting \cite{Dasgupta1987Asymptotically} & $ \mu_{-1} = 1$ and $\mu_k \in [0,1), k \ge 0$ & $\Sigma_k = (\frac{\mu_k}{(1-\mu_k)\mu_{k-1}} - 1)P_k.$ 
\\
5. Exponential Resetting \cite{Lai2023Exponential} & positive-definite $P_\infty \in \BBR^{n \times n}$ & $\Sigma_k = (\lambda P_k^{-1} + (1-\lambda) P_\infty^{-1})^{-1} - P_k.$ 
\\
6. Covariance Resetting \cite{goodwin1983deterministic}& positive-definite $P_\infty \in \BBR^{n \times n}$ and resetting criteria 
& $\Sigma_k = \begin{cases}
        P_\infty - P_k & \textnormal{criteria met}, \\
        0_{n \times n} & \textnormal{otherwise}.
    \end{cases}$ 
\\
\begin{tabular}{@{}l@{}} 7. Directional Forgetting by \\ Information Matrix Decomp. \cite{cao2000directional} \end{tabular}  &  $\lambda \in (0,1]$  & $\Sigma_k = \frac{1-\lambda}{\lambda} C_k^\rmT (C_k P_k^{-1} C_k^\rmT)^{-1}  C_k.$ 
\\
8. Variable-Direction Forgetting \cite{goel2020recursive} & Positive-definite $\Lambda_k \in \BBR^{n \times n}$, $k \ge 0$ & $\Sigma_k = \Lambda_k^{-1} P_k^{-1} \Lambda_k^{-1} - P_k.$ 
\\
\bottomrule[1pt]
\end{tabular}
\label{table:summary}
\end{table*}

\section{Adaptive Kalman Filtering Developed from RLS Forgetting Algorithms}
Thus far, we have shown that various extensions of RLS from the literature are special cases of the Kalman filter, in part, by a special choice of the process noise covariance given by \eqref{eqn: Sigma_k definition}.
Motivated by this relationship, we propose a class of adaptive Kalman filters by replacing the prior covariance update equation \eqref{eqn: KF P predict} with the adaptive prior covariance update equations
\begin{align}
    P_{{\rm forget},k} &= P_k + \Sigma_{{\rm forget},k}, \label{eqn: adaptive prior covariance 1} \\
    P_{k+1 \vert k} &= A_k P_{{\rm forget},k} A_k^\rmT + \Sigma_{{\rm Kalman},k}, \label{eqn: adaptive prior covariance 2}
\end{align}
where, for all $k \ge 0$, positive-semidefinite $\Sigma_{{\rm forget},k} \in \BBR^{n \times n}$ is designed from an extension of RLS (e.g. right column of Table \ref{table:summary}), and positive-semidefinite $\Sigma_{{\rm Kalman},k} \in \BBR^{n \times n}$ is designed using traditional methods of the Kalman filter.

While there are as many variants of this algorithm as extensions of RLS, we select a particular extension of RLS to show the potential benefits in state estimation.

\subsection{Kalman Filter with Robust Variable Forgetting Factor}
Consider the mass-spring-damper system in Figure \ref{fig:mass-spring-damper}, with mass $m = 10$, spring constant $k = 5$, and damping coefficient $c = 3$. 
The vertical displacement of the mass at time $t$ is $z(t)$, where $z = 0$ when the mass is at rest.  
Assume that $z(0) = -1$ and $\dot{z}(0) = 1$
Furthermore, a downward force $F(t) = 10 \sin(t)$ is applied on the mass at time $t$.
This nominal mass-spring-damper system can be modeled as 
\begin{align}
\label{eqn: nominal spring mass damper}
    \begin{bmatrix}
        \dot{z}(t) \\
        \ddot{z}(t)
    \end{bmatrix}
    =
    \begin{bmatrix}
        0 & 1 \\ 
        -\nicefrac{k}{m} & -\nicefrac{b}{m}
    \end{bmatrix}
    \begin{bmatrix}
        z(t) \\
        \dot{z}(t)
    \end{bmatrix}
    + 
    \begin{bmatrix}
        0 \\ \nicefrac{1}{m}
    \end{bmatrix}
    F(t).
\end{align}

However, we additionally assume there is a wall at $z=2$. If the mass collides with this wall, the mass reverses direction with the same speed. 
Note that such intermittent collisions can be interpreted as impulsive disturbances on the system. 
Finally, we consider, for all $k \ge 0$, the measurements $y_k \in \BBR$, given by
\begin{align}
    y_k = z(k T_s) + \dot{z}(k T_s) + v_k,
    \label{eqn: measurement spring-mass-damper}
\end{align}
where $T_s = 0.1$ is the sampling time and, for all $k \ge 0$, $v_k \sim \mathcal{N}(0,\Gamma_k)$ is the measurement noise with $\Gamma_k = 0.01$.
The goal is to estimate the vertical displacement $z(t)$ and velocity $\dot{z}(t)$ without knowledge of the wall at $z=2$.
We will assume that the nominal model \eqref{eqn: nominal spring mass damper} and the measurement noise covariance $\Gamma_k$ are known.
\begin{figure}
    \centering
    \includegraphics[width = .18 \textwidth]{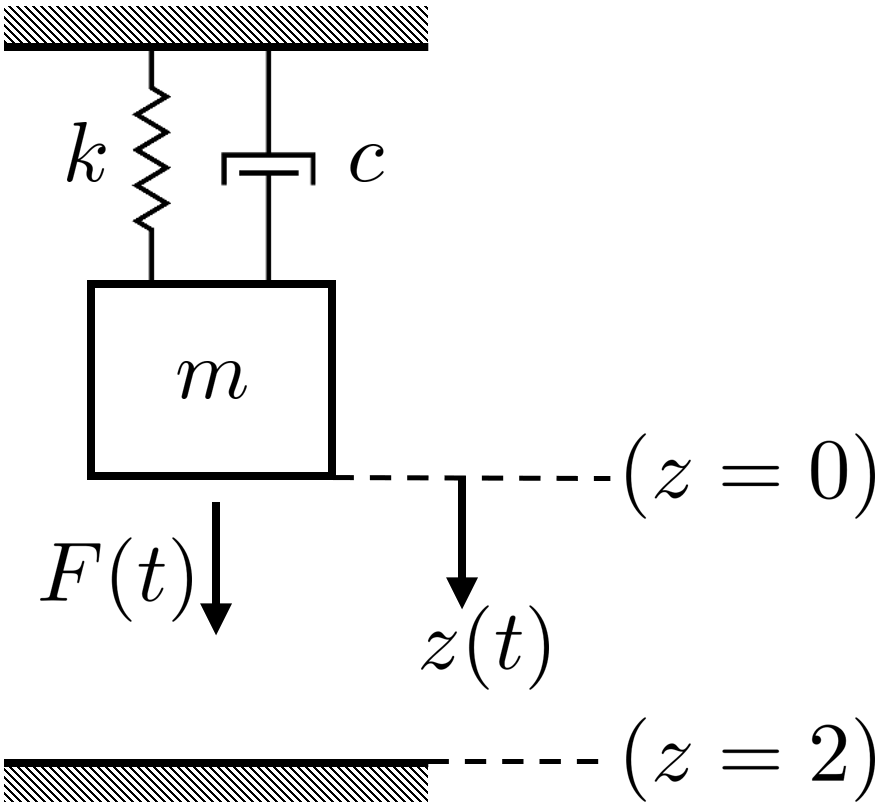}
    \caption{Mass-spring-damper system diagram. The mass can collide with the wall at $z=2$, reversing direction and keeping the same speed.}
    \label{fig:mass-spring-damper}
\end{figure}

We begin by discretizing \eqref{eqn: nominal spring mass damper} and \eqref{eqn: measurement spring-mass-damper} using zero-order hold and sampling time $T_s$ to obtain the nominal system
\begin{align}
    x_{k+1} &= A_k x_k + 
    B_k u_k, \nonumber \\
    y_k &= C_k x_k + v_k,
    \label{eqn: nominal discrete system}
\end{align}
where, for all $k \ge 0$, $x_k \triangleq [z(k T_s) \ \dot{z}(k T_s)  ]^\rmT$, $u_k \triangleq F(k T_s)$,
\begin{align}
    A_k & \triangleq \begin{bmatrix}
        0.9975 & 0.09843\\ -0.04922 & 0.9680
    \end{bmatrix}
    , \
    B_k \triangleq \begin{bmatrix}
        \num{4.948 e -4} \\ \num{9.843 e -3}
    \end{bmatrix}, \nonumber 
    \\
    C_k & \triangleq \begin{bmatrix}
        1 & 1
    \end{bmatrix}.
\end{align}
We first consider state estimation using the Kalman filter with the nominal discrete system \eqref{eqn: nominal discrete system}, and $P_0 = 0.1 I_2$, $\hat{x}_0 = [0 \ 0]^\rmT$, and, for all $k \ge 0$, $\Gamma_k = 0.01$, $\Sigma_k = 0.01 I_2$.

Second, we will consider an adaptive Kalman filter with adaptive prior covariance update equations \eqref{eqn: adaptive prior covariance 1} and \eqref{eqn: adaptive prior covariance 2} developed from variable-rate forgetting \cite{bruce2020convergence}. 
This adaptive Kalman filter also uses the nominal discrete system \eqref{eqn: nominal discrete system}, and $P_0 = 0.1 I_2$, $\hat{x}_0 = [0 \ 0]^\rmT$, and, for all $k \ge 0$, $\Gamma_k = 0.01$.
Then, for all $k \ge 0$, let 
\begin{align}
    \Sigma_{{\rm Kalman},k} = 0.01 I_2, \quad \Sigma_{{\rm forget},k} = (\frac{1}{\lambda_k} - 1) P_k,
\end{align}
where, the forgetting factor $\lambda_k \in (0,1]$ is chosen using the robust variable forgetting factor algorithm developed in \cite{paleologu2008robust}.
We've chosen this algorithm for its ability to improve tracking of impulsive changes of parameters \cite{paleologu2008robust}.
For all $k \ge 0$, let 
\begin{align}
    \hat{\sigma}_{e,k}^2 &= \alpha \hat{\sigma}_{e,k-1}^2 + (1-\alpha) (y_k - C_k \hat{x}_k) ^2 , \\
    \hat{\sigma}_{q,k}^2 &= \alpha \hat{\sigma}_{q,k-1}^2 + (1-\alpha) (\hat{x}_k^\rmT P_k \hat{x}_k) ^2, \\
    \hat{\sigma}_{v,k}^2 &= \beta \hat{\sigma}_{v,k-1}^2 + (1-\beta) (y_k - C_k \hat{x}_k) ^2,
\end{align}
where $\hat{\sigma}_{e,-1} \triangleq \hat{\sigma}_{q,-1} \triangleq \hat{\sigma}_{v,-1} \triangleq 1$, $\alpha \triangleq 1-\frac{1}{K_\alpha n}$, and $\beta \triangleq 1-\frac{1}{K_\beta n}$, where $K_\alpha \triangleq 2$, $K_\beta \triangleq 10$, and $n = 2$ since the system is second-order.

Then, for all $k \ge 0$, if $\hat{\sigma}_{e,k} \le \hat{\sigma}_{v,k}$, $\lambda_k = \lambda_{\rm max}$, otherwise
\begin{align}
    \lambda_k = 
    \max \left\{ \min \left\{ \frac{\hat{\sigma}_{q,k} \hat{\sigma}_{v,k}}{\xi + \vert \hat{\sigma}_{e,k} -\hat{\sigma}_{v,k} \vert } , \lambda_{\rm max} \right\} , \lambda_{\rm min} \right\},
\end{align}
where $\xi \triangleq 10^{-6}$, $\lambda_{\rm min} \triangleq 0.5$, and $\lambda_{\rm max} \triangleq 1$.
For details on robust variable forgetting and tuning of parameters, see \cite{paleologu2008robust}.

Figure \ref{fig: state} shows state estimation of the Kalman filter (KF) and the adaptive Kalman filter (KF*), as well as the forgetting factor $\lambda_k$, all with zero-order hold.
Note that after each of the four collisions the mass makes with the wall at $z=2$, the forgetting factor $\lambda_k$ briefly but drastically decreases. 
This results in improved displacement and velocity estimation immediately after each collision. 
This can be more clearly seen in the error between the true state and the estimated state in Figure \ref{fig: error}.
Figure \ref{fig: covariance} shows $\sigma_{\hat{z}}$ and $\sigma_{\hat{\dot{z}}}$, the diagonal elements of the covariance matrix $P_k$, which can also be interpreted as the marginal variance of the displacement and velocity state estimate, respectively.
Note that in the adaptive Kalman filter (KF*), there is a sudden increase in both $\sigma_{\hat{z}}$ and $\sigma_{\hat{\dot{z}}}$ after each collision, allowing for quicker adaptation.

\begin{figure}[ht]
    \centering
    \includegraphics[width = .48 \textwidth]{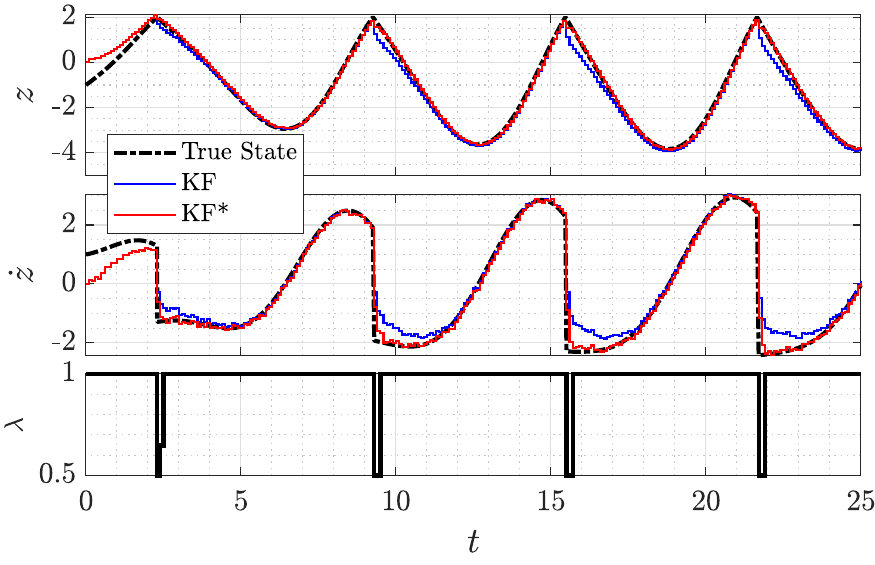}
    \caption{Vertical displacement ($z$) and velocity ($\dot{z}$) estimation using Kalman filter (KF) and adaptive Kalman filter (KF*). $\lambda$ shows the forgetting factor used in KF*.}
    \label{fig: state}
\end{figure}

\begin{figure}[ht]
    \centering
    \includegraphics[width = .48 \textwidth]{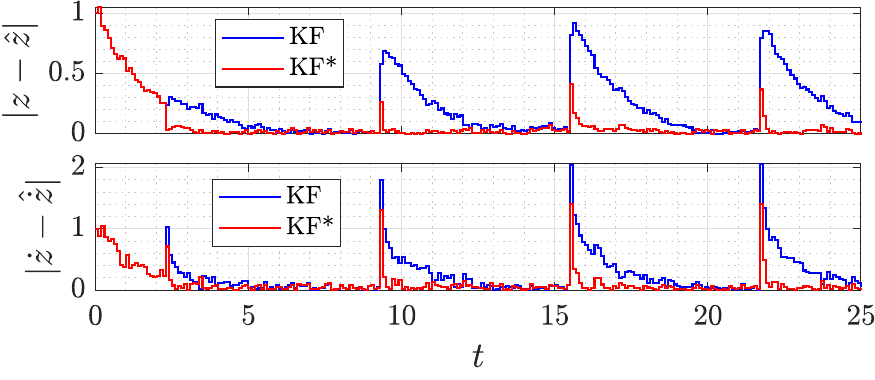}
    \caption{Estimation error of vertical displacement (top) and velocity (bottom) using Kalman filter (KF) and adaptive Kalman filter (KF*).}
    \label{fig: error}
\end{figure}

\begin{figure}[ht]
    \centering
    \includegraphics[width = .48 \textwidth]{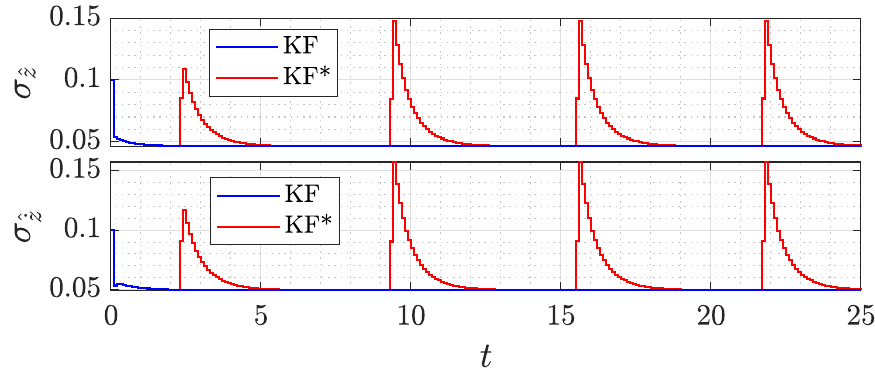}
    \caption{Marginal variance of vertical displacement (top) and velocity (bottom) using Kalman filter (KF) and adaptive Kalman filter (KF*).}
    \label{fig: covariance}
\end{figure}

\section{Conclusion}

This work presents the Kalman filter least squares cost function whose recursive minimizer gives the Kalman filter update equations. 
An important consequence of this cost function is that various extensions of RLS from the literature are special cases of the Kalman filter.
Motivated by this result, we propose a new adaptive Kalman filters, whose prior covariance update is modified to include RLS forgetting.
While the numerical example we presented shows the potential benefits of adaptive Kalman filtering with robust variable forgetting factor \cite{paleologu2008robust} in the presence of impulsive disturbances, there are numerous other forgetting algorithms in the RLS literature to be considered (several summarized in Table \ref{table:summary}).
Future work includes further exploration into how, and in what situations, such extensions may be beneficial.



\bibliographystyle{IEEEtran}
\bibliography{refs}

\begingroup
\allowdisplaybreaks
\section*{Appendix}
\begin{lem}
\label{lem: matrix inversion lemma}
Let $A \in \BBR^{n \times n}$, $U \in \BBR^{n \times p}$, $C \in \BBR^{p \times p}$, $V \in \BBR^{p \times n}$. Assume $A$, $C$, and $A+UCV$ are nonsingular. Then, $(A+UCV)^{-1} = A^{-1} - A^{-1}U(C^{-1} + VA^{-1} U)^{-1} V A^{-1}$.
\end{lem}

\begin{lem}
\label{lem: quadratic cost minimizer}
Let $A \in \BBR^{n \times n}$ be positive definite, let $b \in \BBR^n$ and $c \in \BBR$, and define $f\colon\BBR^n \rightarrow \BBR$ by $f(x) \triangleq x^\rmT A x + 2 b^\rmT x + c$. 
Then, $f$ has a unique stationary point, which is the global minimizer given by $\argmin_{x \in \BBR^n} f(x) = -A^{-1} b$.
\end{lem}

\begin{proof}[Proof of Theorem \ref{theo: KFRLS}]
We write $J_0(\hat{x})$, given by \eqref{eqn: KFRLS cost}, as
\begin{align*}
    & J_0(\hat{x}) = \Vert y_0 - C_0\hat{x} \Vert_{\Gamma_0^{-1}}^2 + \Vert A_0^{-1}(\hat{x} - B_0 u_0) - \hat{x}_0 \Vert_{P_0^{-1}-F_0}^2 , \\
    & = \Vert y_0 - C_0\hat{x} \Vert_{\Gamma_0^{-1}}^2 + \Vert \hat{x} - (A_0\hat{x}_0 + B_0 u_0) \Vert_{A_0^{-\rmT} (P_0^{-1}-F_0) A_0^{-1}}^2
\end{align*}
Next, we expand $J_0(\hat{x}) = \hat{x}^\rmT H_0 \hat{x} + 2 b_0^\rmT \hat{x} + c_0$, where
\begin{align}
    H_0 &\triangleq  C_0^\rmT \Gamma_0^{-1} C_0 + A_0^{-\rmT} (P_0^{-1} - F_0) A_0^{-1} \label{eqn: H0 defn} \\
    b_0 &\triangleq  -C_0^\rmT \Gamma_0^{-1} y_0 - A_0^{-\rmT} (P_0^{-1} - F_0) A_0^{-1} (A_0 \hat{x}_0 + B_0 u_0) \nonumber \\
    c_0 &\triangleq y_0^\rmT \Gamma_0^{-1} y_0 + \Vert A_0^{-1}(A_0\hat{x}_0 + B_0 u_0) \Vert _{P_0^{-1}-F_0}^2  \nonumber
\end{align}
Defining $P_1 \triangleq H_0^{-1}$, it follows from the expression of $H_0$ that \eqref{eqn: KFRLS Pinv Update} is satisfied for $k = 0$. 
Furthermore, if $k = 0$, \eqref{eqn: Pinv - F is pos def} simplifies to $F_0 \prec P_0^{-1}$ and hence $P_0^{-1} - F_0 \succ 0$. 
Furthermore, since $\Gamma_0$ is positive definite, it follows from definition \eqref{eqn: H0 defn} that $H_0$ is positive definite. 
Therefore, from Lemma \ref{lem: quadratic cost minimizer}, it follows that the unique minimizer $\hat{x}_1$ of $J_0$ is given by $\hat{x}_1 = - H_0^{-1} b_0$, which can be expanded as
\begin{align*}
    \hat{x}_1 =& P_1 \big[ C_0^\rmT \Gamma_0^{-1} y_0 + A_0^{-\rmT} (P_0^{-1} - F_0) A_0^{-1} (A_0 \hat{x}_0 + B_0 u_0) \big] \\
    =& P_1 \big[ C_0^\rmT \Gamma_0^{-1} y_0  -  C_0^\rmT \Gamma_0^{-1} C_0 (A_0 \hat{x}_0 + B_0 u_0) 
    \\
    & + (A_0^{-\rmT} (P_0^{-1} - F_0) A_0^{-1} + C_0^\rmT \Gamma_0^{-1} C_0) (A_0 \hat{x}_0 + B_0 u_0)\big] \\
    =& P_1 \big[ C_0^\rmT \Gamma_0^{-1} y_0 -  C_0^\rmT \Gamma_0^{-1} C_0 (A_0 \hat{x}_0 + B_0 u_0) 
    \\
    & + P_1^{-1} (A_0 \hat{x}_0 + B_0 u_0)\big] 
    \\
    =& A_0 \hat{x}_0 + B_0 u_0 + P_1 C_0^\rmT \Gamma_0^{-1}(y_0 - C_0 (A_0 \hat{x}_0 + B_0 u_0)).
\end{align*}
Hence, \eqref{eqn: KFRLS x Update} holds for $k = 0$. Moreover, since $H_0$ is positive definite, it follows that $P_1 = H_0^{-1}$ is also positive definite.

Now, let $k \ge 1$. 
%
%
Note that for all $0 \le i \le k$,
\begin{align*}
    (\mathcal{T}_{i,k+1}(\hat{x}) - \hat{x}_i)
    &= \Phi_{i,k+1}(\hat{x} - \mathcal{B}_{k+1,i} \mathcal{U}_{k+1,i}) - \hat{x}_i 
    \\
    &= \Phi_{i,k+1}(\hat{x} - (\Phi_{k+1,i} \hat{x}_i + \mathcal{B}_{k+1,i} \mathcal{U}_{k+1,i})) 
    \\
    &= \Phi_{i,k+1}(\hat{x} - \mathcal{T}_{k+1,i}(\hat{x}_i)).
\end{align*}
Hence, $J_k(\hat{x})$, given by \eqref{eqn: KFRLS cost}, can be written as
%
%
\begin{align*}
    & J_k(\hat{x}) = \Vert \Phi_{0,k+1} (\hat{x} - \mathcal{T}_{k+1,0}(\hat{x}_0)) \Vert_{P_0^{-1}}^2 + 
    \\
    & \sum_{i=0}^{k}  \Vert y_i - C_i \mathcal{T}_{i+1,k+1}(\hat{x}) \Vert _{\Gamma_i^{-1}}^2 
    -  \Vert \Phi_{i,k+1} ( \hat{x} - \mathcal{T}_{k+1,i}(\hat{x}_i) ) \Vert_{F_i} .
\end{align*}
Next, we expand $J_k(\hat{x}) = \hat{x}^\rmT H_k \hat{x} + 2 b_k^\rmT \hat{x} + c_k$, where
\begin{align}
    & H_k \triangleq \sum_{i=0}^k \Big[ \Phi_{i+1,k+1}^\rmT C_i^\rmT \Gamma_i^{-1} C_i \Phi_{i+1,k+1} 
    - \Phi_{i,k+1}^\rmT F_i \Phi_{i,k+1} \Big] \nonumber
    \\
    & \hspace{10pt} + \Phi_{0,k+1}^\rmT P_0^{-1} \Phi_{0,k+1}, \label{eqn: H_k defn} \\
    & b_k \triangleq \sum_{i=0}^k \Big[  \Phi_{i,k+1}^\rmT F_i \Phi_{i,k+1} \mathcal{T}_{k+1,i}(\hat{x}_i) \nonumber
    \\
    & \hspace{10pt} -\Phi_{i+1,k+1}^\rmT C_i^\rmT \Gamma_i^{-1} (y_i + C_i \Phi_{i+1,k+1} \mathcal{B}_{k+1,i+1}  \mathcal{U}_{k+1,i+1}) \Big]  \nonumber \\
    & \hspace{10pt} - \Phi_{0,k+1}^\rmT P_0^{-1} \Phi_{0,k+1} \mathcal{T}_{k+1,0}(\hat{x}_0), \label{eqn: b_k defn} \\
    & c_k \triangleq  \sum_{i=0}^k \Big[ \Vert y_i + C_i \Phi_{i+1,k+1} \mathcal{B}_{k+1,i+1} \mathcal{U}_{k+1,i+1} \Vert_{\Gamma_i}^2 \nonumber
    \\
    & \hspace{10pt}  - \Vert \Phi_{i,k+1} \mathcal{T}_{k+1,i}(\hat{x}_i) \Vert_{ F_i }^2 \Big] \nonumber
    %
    +  \Vert \Phi_{0,k+1} \mathcal{T}_{k+1,0}(\hat{x}_i) \Vert_{P_0^{-1}}^2. \nonumber
\end{align}
Note that $H_k$ can be written recursively as
\begin{align}
    H_k =& \sum_{i=0}^{k-1} \left[ \Phi_{i+1,k+1}^\rmT C_i^\rmT \Gamma_i^{-1} C_i \Phi_{i+1,k+1} 
    - \Phi_{i,k+1}^\rmT F_i \Phi_{i,k+1} \right] \nonumber
    \\
    & + \Phi_{0,k+1}^\rmT P_0^{-1} \Phi_{0,k+1}  
    + C_k^\rmT \Gamma_k^{-1} C_k - A_k^{-\rmT} F_k A^{-1}  \nonumber
    \\
    =& A_k^{-\rmT} \Big[ \sum_{i=0}^{k-1} \left( \Phi_{i+1,k}^\rmT C_i^\rmT \Gamma_i^{-1} C_i \Phi_{i+1,k} 
    - \Phi_{i,k}^\rmT F_i \Phi_{i,k} \right) \nonumber
    \\
    & + \Phi_{0,k}^\rmT P_0^{-1} \Phi_{0,k} \Big] A_k^{-1}
    + C_k^\rmT \Gamma_k^{-1} C_k - A_k^{-\rmT} F_k A^{-1} \nonumber
    \\
    =& A_k^{-\rmT} (H_{k-1} - F_k) A_k^{-1} + C_k^\rmT \Gamma_k^{-1} C_k . \label{eqn: Hk update}
\end{align}
Defining $P_{k+1} \triangleq H_k^{-1}$,
it follows that \eqref{eqn: KFRLS Pinv Update} is satisfied. 

Next, to write a recursive update for $b_k$, we first write $b_k$ as $b_k = b_{k,1} + b_{k,2} + b_{k,3}$, where
\begin{align*}
    & b_{k,1} \triangleq \sum_{i=0}^{k-1} \Big[  \Phi_{i,k+1}^\rmT F_i \Phi_{i,k+1} \mathcal{T}_{k+1,i}(\hat{x}_i) \nonumber
    \\
    & \hspace{10pt} -\Phi_{i+1,k+1}^\rmT C_i^\rmT \Gamma_i^{-1} (y_i + C_i \Phi_{i+1,k+1} \mathcal{B}_{k+1,i+1}  \mathcal{U}_{k+1,i+1}) \Big],
    \\
    & b_{k,2} \triangleq - \Phi_{0,k+1}^\rmT P_0^{-1} \Phi_{0,k+1} \mathcal{T}_{k+1,0}(\hat{x}_0) 
    \\
    & b_{k,3} \triangleq A_k^{-\rmT} F_k A_k^{-1} (A_k \hat{x}_k + B_k u_k) - C_k^\rmT \Gamma_k^{-1} y_k.
\end{align*}
Note that $b_{k,1}$ is the sum of first $k-1$ terms of summation in \eqref{eqn: b_k defn}, $b_{k,3}$ is the last term of summation in \eqref{eqn: b_k defn}, and $b_{k,2}$ is the remaining term outside the summation. Next, note, for all $0 \le i \le k$, the identities
\begin{align}
    & \Phi_{i,k+1} \mathcal{T}_{k+1,i}(\hat{x}_i) 
    = \Phi_{i,k} \mathcal{T}_{k,i}(\hat{x}_i) + \Phi_{i,k+1} B_k u_k, \label{eqn: kfrls proof identity 1} \\
    & \Phi_{i+1,k+1} \mathcal{B}_{k+1,i+1}  \mathcal{U}_{k+1,i+1} \nonumber
    \\
    & \hspace{30pt} = \Phi_{i+1,k} \mathcal{B}_{k,i+1} \mathcal{U}_{k,i+1} + \Phi_{i+1,k+1} B_k u_k. \label{eqn: kfrls proof identity 2}
\end{align}
Using \eqref{eqn: kfrls proof identity 1} and \eqref{eqn: kfrls proof identity 2}, we can write $b_{k,1} = b_{k,1,1} + b_{k,1,2}$, where
\begin{align*}
    b_{k,1,1} & \triangleq  \sum_{i=0}^{k-1} \big[ 
    \Phi_{i,k+1}^\rmT F_i \Phi_{i,k} \mathcal{T}_{k,i}(\hat{x}_i) 
    \\
    & \hspace{20pt} -\Phi_{i+1,k+1}^\rmT C_i^\rmT \Gamma_i^{-1} (y_i + C_i \Phi_{i+1,k} \mathcal{B}_{k,i+1} \mathcal{U}_{k,i+1}) \big],
    \\
    b_{k,1,2} &\triangleq  \sum_{i=0}^{k-1} \big[  \Phi_{i,k+1}^\rmT F_i \Phi_{i,k+1} B_k u_k
    \\
    & \hspace{20pt} -\Phi_{i+1,k+1}^\rmT C_i^\rmT \Gamma_i^{-1} C_i \Phi_{i+1,k+1} B_k u_k   \big].
\end{align*}
Similarly, \eqref{eqn: kfrls proof identity 1} implies that $b_{k,2} = b_{k,2,1} + b_{k,2,2}$, where
\begin{align*}
    b_{k,2,1} \triangleq   -\Phi_{0,k+1}^\rmT P_0^{-1} \Phi_{0,k} \mathcal{T}_{k,0}(\hat{x}_0) ,
    \\
    b_{k,2,2} \triangleq  - \Phi_{0,k+1}^\rmT P_0^{-1} \Phi_{0,k+1} B_k u_k  .
\end{align*}
It then follows from \eqref{eqn: b_k defn} and \eqref{eqn: H_k defn}, respectively, that
\begin{align*}
    b_{k,1,1} + b_{k,2,1} &= A_k^{-\rmT} b_{k-1}, \\
    b_{k,1,2} + b_{k,2,2} &= - A_k^{-\rmT} H_{k-1} A_k^{-1} B_k u_k .
\end{align*}
Hence, we obtain the recursive update
\begin{align*}
    b_k =& A_k^{-\rmT} b_{k-1} - A_k^{-\rmT} H_{k-1} A_k^{-1} B_k u_k 
    \\
    &+ A_k^{-\rmT} F_k A_k^{-1} (A_k \hat{x}_k + B_k u_k) - C_k^\rmT \Gamma_k^{-1} y_k.
\end{align*}

Finally, note that \eqref{eqn: H_k defn} can be used to rewrite \eqref{eqn: Pinv - F is pos def} as $F_k \prec H_{k-1}$, and hence $H_{k-1} - F_k \succ 0$.
Furthermore, since $\Gamma_k$ is positive definite, it follows from \eqref{eqn: Hk update} that $H_k$ is positive definite. 
Therefore, by Lemma \ref{lem: quadratic cost minimizer}, the unique minimizer $\hat{x}_{k+1}$ of $J_k$ is given by $\hat{x}_{k+1} = - H_k^{-1} b_k$, which simplifies to
\begin{align*}
    & \hat{x}_{k+1} = - H_k^{-1} b_k = -P_{k+1} b_k 
    \\
    &= P_{k+1} \big[ -A_k^{-\rmT} b_{k-1} + A_k^{-\rmT} H_{k-1} A_k^{-1} B_k u_k 
    \\
    & \hspace{40pt} -A_k^{-\rmT} F_k A_k^{-1} (A_k \hat{x}_k + B_k u_k) + C_k^\rmT \Gamma_k^{-1} y_k \big]  
    \\
    &= P_{k+1} \big[A_k^{-\rmT} P_k^{-1} \hat{x}_k + A_k^{-\rmT} P_{k}^{-1} A_k^{-1} B_k u_k 
    \\
    & \hspace{40pt} - A_k^{-\rmT} F_k A_k^{-1} (A_k \hat{x}_k + B_k u_k) + C_k^\rmT \Gamma_k^{-1} y_k \big] 
    \\
    &= P_{k+1} \big[ A_k^{-\rmT} (P_k^{-1} - F_k) A_k^{-1} (A_k \hat{x}_k + B_k u_k) + C_k^\rmT \Gamma_k^{-1} y_k \big] 
    \\
    &= P_{k+1} \big[ C_k^\rmT \Gamma_k^{-1} y_k - C_k^\rmT \Gamma_k^{-1} C_k(A_k \hat{x}_k + B_k u_k)
    \\
    & \hspace{10pt} + \left(A_k^{-\rmT} (P_k^{-1} - F_k) A_k^{-1} + C_k^\rmT \Gamma_k^{-1} C_k \right)(A_k \hat{x}_k + B_k u_k) \big] 
    \\
    &= P_{k+1} \big[C_k^\rmT \Gamma_k^{-1} (y_k - C_k(A_k \hat{x}_k + B_k u_k))
    \\
    & \hspace{40pt} + P_{k+1}^{-1}(A_k \hat{x}_k + B_k u_k)  \big] 
    \\
    &= A_k \hat{x}_k + B_k u_k + P_{k+1} C_k^\rmT \Gamma_k^{-1} (y_k - C_k(A_k \hat{x}_k + B_k u_k) ).
\end{align*}
Hence, \eqref{eqn: KFRLS x Update} is satisfied. Finally, since $H_k$ is positive definite, it follows that $P_{k+1} = H_k^{-1}$ is also positive definite.
\end{proof}

\endgroup

\end{document}